\documentclass[11pt, oneside]{article}   	% use "amsart" instead of "article" for AMSLaTeX format
\usepackage{geometry}                		% See geometry.pdf to learn the layout \mathrm{OPT}ions. There are lots.
\usepackage{graphicx}				% Use pdf, png, jpg, or eps§ with pdflatex; use eps in DVI mode
\usepackage{amssymb}

\usepackage{amssymb}
\usepackage{float}
\usepackage{amssymb}
\usepackage{amsmath}
\usepackage{amsthm}
\usepackage{color}
\usepackage{graphicx}
\usepackage{hyperref}
\usepackage{fullpage}
\usepackage{graphicx}
\usepackage{mathtools}
\usepackage{microtype}
\usepackage{bbm}
\usepackage{booktabs}

\newtheorem{theorem}{Theorem}
\newtheorem{lemma}{Lemma}
\newtheorem{definition}{Definition}

\DeclareMathOperator*{\argmax}{argmax}

\title{A Note On Deterministic Submodular Maximization With Bounded Curvature}
\author{Wenxin Li\\
The Ohio State University\\
{\tt li.7328@osu.edu}\\
{\tt wenxinliwx.1@gmail.com}\\	
}
\begin{document}
\maketitle
\begin{abstract}
We show that the recent breakthrough result of Buchbinder and Feldman \cite{buchbinder2024deterministic} could further lead to a deterministic $(1-\kappa_{f}/e-\varepsilon)$-approximate algorithm for maximizing a submodular function with curvature $\kappa_{f}$ under matroid constraint.
\end{abstract}
%non-decreasing
A set function $f:2^{E}\rightarrow \mathbb{R}^{+}$ defined on ground $E$ of size $n$ is \emph{submodular}, if inequality $f(S)+f(T)\geq f(S\cup T)+f(S\cap T)$ holds for any two subsets $S,T\subseteq E$. One of the important concepts that influences the performance and analysis of algorithms for monotone submodular maximization is curvature \cite{conforti1984submodular}. The curvature of a submodular function measures how far the function deviates from being modular (i.e., linear), and it plays a crucial role in determining the approximation guarantees of various algorithms. Sviridenko et al. \cite{sviridenko2017optimal} presented two optimal algorithms achieving an approximation ratio of $(1-\kappa_{f}/e-\varepsilon)$, by optimizing the sum of a submodular and a linear function.

Several studies \cite{buchbinder2018deterministic, buchbinder2023deterministic, li2022submodular} have focused on the development and analysis of deterministic algorithms for submodular maximization, offering provable guarantees and practical effectiveness. Recently, Buchbinder and Feldman \cite{buchbinder2024deterministic} achieved a significant breakthrough by introducing the first deterministic algorithm with an optimal approximation ratio for submodular function maximization under matroid constraint. We demonstrate that if (approximate) local maxima are used as solutions in the non-oblivious local search, the algorithm also achieve an optimal approximation ratio for submodular functions with bounded curvature. 

The main result is stated in the following theorem. It is important to emphasize that this result is attributed to \cite{buchbinder2024deterministic}, and is also inspired by the work presented in \cite{sviridenko2017optimal}. Most of our notations are aligned with that of \cite{buchbinder2024deterministic}, our goal here is to point out this fact.

\begin{theorem}\label{maintheorem}
There exists a deterministic polynomial time algorithm that achieves an approximation ratio of $(1-\kappa_{f}/e-\varepsilon)$ for maximizing a submodular function $f(\cdot)$ with bounded curvature $\kappa_{f}$ under the matroid constraint.
\end{theorem}

Firstly, we introduce some basic definitions that will be used throughout the paper.
\begin{definition}[Curvature \cite{sviridenko2017optimal}] The curvature of a monotone submodular $f$ is defined as,
\begin{align*}
\kappa_{f} = 1- \min_{S, e\notin S}\frac{f(S+e)-f(e)}{f(e)}
\end{align*}
\end{definition}

\begin{definition}[Submodular Ratio] For a set function $f(\cdot)$, its submodularity ratio is defined as 
\begin{align*}
\gamma_{f} = \min_{T\subseteq S, e\notin S} \frac{f(T+e)-f(T)}{f(S+e)-f(S)}
\end{align*}
$f(\cdot)$ is submodular iff $\gamma_{f}=1$.
\end{definition}

The following lemma is implied by \cite{buchbinder2024deterministic}, we include it here for completeness.
\begin{lemma}[\cite{buchbinder2024deterministic}]
Following the notations in \cite{buchbinder2024deterministic}, and let $\mathrm{OPT}=\argmax\{f(T)|T\in \mathcal{I}\}$, $\alpha_{i} = \frac{(1+ \frac{\gamma^{3}_{g}}{\ell})^{i-1}}{\binom{\ell-1}{i-1}}$. Consider submodular function $\Phi_{g}(\cdot): 2^{E}\times [\ell] \rightarrow \mathbb{R}^{+}$ defined as
\begin{align*}
\Phi_{g}(S)=\frac{\gamma_{g}}{\ell}(1 +  \frac{\gamma^{3}_{g}}{\ell})^{-\ell}\cdot\sum_{J\subseteq [\ell]}\alpha_{|J|} \cdot g(\pi_{J}(S)),
\end{align*}
for submodular function $g(\cdot): 2^{E}\rightarrow \mathbb{R}^{+}$, we have 
\begin{align}\label{submodularpart}
g(\pi_{[\ell]}(S))  \geq & (1 - (1 + \frac{\gamma^{3}_{g}}{\ell})^{-\ell} )\cdot g(\mathrm{OPT})+  [\Phi_{g}(S)-\frac{1}{\ell}\sum_{j\in[\ell]}\sum_{(u, k) \in S} \Phi_{g}(S-(u, k)+h_{j}((u,k))) ]
\end{align}
where $h_{j}(\cdot):S\rightarrow \mathrm{OPT}\times\{j\}$ is the bijection ensured by basis exchange property of $\mathcal{I}^{\prime}$. 
%satisfies that $S-(u, k)+h_{j}((u,k))\in \mathcal{I}^{\prime}$.
\end{lemma}

\begin{proof}
% We first claim that $\gamma_{\Phi_{g}}\geq \gamma_{g}$, since for any $T\subseteq S$ and $(e,i)\notin S$,
% \begin{align*}
% \Phi_{g}((e,i)|T) = \sum_{J\subseteq [\ell], i\notin J}\alpha_{|J|} \cdot g(e|\pi_{J}(T))\geq \gamma_{g}\cdot \sum_{J\subseteq [\ell], i\notin J}\alpha_{|J|} \cdot g(e|\pi_{J}(S))=\gamma_{g}\cdot\Phi_{g}((e,i)|S).
% \end{align*}

% \begin{align*}
% &\sum_{u \in \pi_J(S)\cap \mathrm{OPT}}\sum_{J \subseteq [\ell], j \in J} \alpha_{|J|} \cdot g(h(u) | \pi_J(S))-\sum_{u \in \pi_J(S)\cap \mathrm{OPT}}\sum_{J \subseteq [\ell], k \in J} \alpha_{|J|} \cdot g(u | \pi_J(S)-u)\\
% =& \sum_{u \in \pi_J(S)\cap \mathrm{OPT}}\sum_{J \subseteq [\ell], j \in J} \alpha_{|J|} \cdot g(u | \pi_J(S))\\
% =& \sum_{(u,k) \in S, u \in \pi_J(S)\cap \mathrm{OPT}} \sum_{J \subseteq [\ell], j \in J} \alpha_{|J|} \cdot g(\pi_J(S))\\
% \end{align*}
For every $j \in [\ell]$, we have $\mathrm{OPT} \times \{j\} \in \mathcal{I}'$. Consider element $u\in \pi_{J}(S)\cap \mathrm{OPT}$ and $(u, k)\in S$, we have
\begin{align}\label{ineq1}
&\frac{\ell}{\gamma_{g}}(1 +  \frac{\gamma^{3}_{g}}{\ell})^{\ell}[\Phi_{g}(S+(h_{j}(u),j)-(u,k))-\Phi_{g}(S)] \notag\\
= &\frac{\ell}{\gamma_{g}}(1 +  \frac{\gamma^{3}_{g}}{\ell})^{\ell}[\Phi_{g}(S+(u,j)-(u,k))-\Phi_{g}(S)]\notag\\
=&\sum_{J\subseteq [\ell], j\in J, k\notin J}\alpha_{|J|} \cdot g(u|\pi_{J}(S))-\sum_{J\subseteq [\ell], j\notin J, k\in J}\alpha_{|J|} \cdot g(u|\pi_{J}(S)-u)\notag\\
\geq & \sum_{J\subseteq [\ell], j\in J, k\notin J}\alpha_{|J|} \cdot g(u|\pi_{J}(S))-\sum_{J\subseteq [\ell], k\in J}\alpha_{|J|} \cdot g(u|\pi_{J}(S)-u)\notag\\
= & \sum_{J\subseteq [\ell], j\in J}\alpha_{|J|} \cdot g(u|\pi_{J}(S))-\sum_{J\subseteq [\ell], k\in J}\alpha_{|J|} \cdot g(u|\pi_{J}(S)-u)\notag\\
\geq & \gamma_{g}\sum_{J\subseteq [\ell], j\in J}\alpha_{|J|} \cdot g(u|\pi_{J}(S))-\sum_{J\subseteq [\ell], k\in J}\alpha_{|J|} \cdot g(u|\pi_{J}(S)-u)
\end{align}
The last inequality holds because $\gamma_{g}\leq 1$ and $g(\cdot)$ is monotone non-decreasing. For $u\in \pi_{J}(S)\setminus \mathrm{OPT}$ and $(u,k)\in S$,
\begin{align}\label{ineq2}
&\frac{\ell}{\gamma_{g}}(1 +  \frac{\gamma^{3}_{g}}{\ell})^{\ell}[\Phi_{g}(S+(h_{j}(u),j)-(u,k))-\Phi_{g}(S)]\notag\\
=& \sum_{J\subseteq [\ell], k \not \in J, j \in J}  \alpha_{|J|} \cdot g(h_{j}(u) | \pi_{J}(S))+ \sum_{J\subseteq [\ell], k \in J, j \in J} \alpha_{|J|} \cdot [g(\pi_{J}(S)-u+h_{j}(u))-g(\pi_{J}(S))]\notag\\
&- \sum_{J\subseteq [\ell], k \in J, j \not \in J} \alpha_{|J|} \cdot g(u | \pi_{J}(S)-u)\notag\\
\geq & \gamma_{g} \sum_{J\subseteq [\ell],  j \in J}  \alpha_{|J|} \cdot g(h_{j}(u) | \pi_{J}(S)) - \sum_{J\subseteq [\ell], k \in J} \alpha_{|J|} \cdot g(u | \pi_{J}(S)-u),
\end{align}
% According to the definition of $\gamma_{\Phi_{g}}$, we have
where the inequality is due to the fact that
\begin{align*}
&g(\pi_{J}(S)-u+h_{j}(u))-g(\pi_{J}(S))\\
=& g(h_{j}(u)|\pi_{J}(S)-u)-g(u|\pi_{J}(S)-u)\\
\geq &\gamma_{g}\cdot g(h_{j}(u)|\pi_{J}(S))-g(u|\pi_{J}(S)-u).
\end{align*}
Combining (\ref{ineq1}) and (\ref{ineq2}), 
\begin{align*}
&\sum_{v \in \mathrm{OPT}} \sum_{J \subseteq [\ell], j \in J} \alpha_{|J|} \cdot g(v | \pi_J(S))\\
=&\sum_{u \in \pi_J(S)\cap \mathrm{OPT}}\sum_{J \subseteq [\ell], j \in J} \alpha_{|J|} \cdot g(u | \pi_J(S))+\sum_{u \in \pi_J(S)\setminus \mathrm{OPT}}\sum_{J \subseteq [\ell], j \in J} \alpha_{|J|} \cdot g(h_{j}(u) | \pi_J(S))\\
\leq & \frac{1}{\gamma_{g}} \sum_{u \in \pi_J(S)}\sum_{J\subseteq [\ell], k \in J} \alpha_{|J|} \cdot g(u | \pi_{J}(S)-u) +\frac{\ell(1 +  \frac{\gamma^{3}_{g}}{\ell})^{\ell}}{\gamma^{2}_{g}} \sum_{\substack{ (u,k)\in S}}[\Phi_{g}(S+(h_{j}(u),j)-(u,k))-\Phi_{g}(S)]
% &+\frac{1}{\gamma_{g}} \sum_{\substack{u \in \pi_J(S)\setminus \mathrm{OPT}\\ (u,k)\in S}}[\Phi_{g}(S+(v,j)-(u,k))-\Phi_{g}(S)]
\end{align*}
Taking the average of this inequality across all $j \in [\ell]$ results in
\begin{align*}
&\sum_{(u, k) \in S} \sum_{J \subseteq [\ell], k \in J}   \alpha_{|J|} \cdot g(u | \pi_J(S)-u)\\
&\geq \gamma_{g}
\sum_{u \in \mathrm{OPT}} \sum_{J \subseteq [\ell]} \frac{|J|}{\ell} \cdot \alpha_{|J|} \cdot g(u | \pi_J(S))+ \frac{\ell}{\gamma_{g}} (1 +  \frac{\gamma^{3}_{g}}{\ell})^{\ell} [\Phi_{g}(S) - \frac{1}{\ell}\sum_{j\in[\ell]}\sum_{(u, k) \in S} \Phi_{g}(S-(u, k)+h_{j}((u,k))) ].
\end{align*}
Furthermore, 
\begin{align*}
&\sum_{J \subseteq [\ell]} 
\left[\alpha_{|J|} \cdot |J| -\alpha_{|J|+1} \cdot (\ell-|J|)\right] \cdot g(\pi_J(S))\\
=& 
\sum_{k \in [\ell]} \sum_{J \subseteq [\ell], k \in J}  
\alpha_{|J|} \cdot [g(\pi_J(S)) -g(\pi_{J\setminus \{k\}}(S))]\\
\geq & \gamma_{g}\cdot \sum_{k \in [\ell]} \sum_{J \subseteq [\ell], k \in J}    \alpha_{|J|} \cdot  (\sum_{(u, k)\in S} g(u | \pi_J(S)-u) ) \\%\label{ineq-sub1}
= & \gamma_{g}\cdot\sum_{(u, k)\in S} \sum_{J \subseteq [\ell], k \in J}    \alpha_{|J|} \cdot g(u | \pi_J(S)-u)\\
\geq & \gamma^{2}_{g}\cdot \sum_{u \in \mathrm{OPT}} \sum_{J \subseteq [\ell]} \frac{|J|}{\ell} \cdot \alpha_{|J|} \cdot g(u | \pi_J(S)) + \ell(1 +  \frac{\gamma^{3}_{g}}{\ell})^{\ell}\cdot  [\Phi_{g}(S)-\frac{1}{\ell}\sum_{j\in[\ell]}\sum_{(u, k) \in S} \Phi_{g}(S-(u, k)+h_{j}((u,k))) ] \\
\geq & \gamma^{3}_{g}\cdot\sum_{J \subseteq [\ell]}\frac{|J|}{\ell} \cdot \alpha_{|J|} \cdot g(\mathrm{OPT} | \pi_J(S)) +\ell(1 +  \frac{\gamma^{3}_{g}}{\ell})^{\ell}\cdot  [\Phi_{g}(S)-\frac{1}{\ell}\sum_{j\in[\ell]}\sum_{(u, k) \in S} \Phi_{g}(S-(u, k)+h_{j}((u,k))) ] \\ 
\geq & \gamma^{3}_{g}\cdot \sum_{J \subseteq [\ell]}\frac{|J|}{\ell} \cdot \alpha_{|J|} \cdot \left[g(\mathrm{OPT}) -g(\pi_J(S))\right] +\ell(1 +  \frac{\gamma^{3}_{g}}{\ell})^{\ell}\cdot  [\Phi_{g}(S)-\frac{1}{\ell}\sum_{j\in[\ell]}\sum_{(u, k) \in S} \Phi_{g}(S-(u, k)+h_{j}((u,k))) ].
\end{align*}
Since $\alpha_{i+1} (\ell - i) =  (1+ \frac{\gamma^{3}_{g}}{\ell} ) \cdot i\cdot \alpha_i$, the above inequality can be reduced to
\begin{align*} 
	&\alpha_\ell \cdot (\ell + \gamma^{3}_{g}) \cdot g(\pi_{[\ell]}(S)) 
  \\ \geq &
  \ell \cdot [(1 + \gamma^{3}_{g}/\ell)^\ell - 1] \cdot g(\mathrm{OPT})
  +\ell(1 +  \frac{\gamma^{3}_{g}}{\ell})^{\ell}\cdot  [\Phi_{g}(S)-\frac{1}{\ell}\sum_{j\in[\ell]}\sum_{(u, k) \in S} \Phi_{g}(S-(u, k)+h_{j}((u,k))) ],
\end{align*}
which is equivalent to 
\begin{align*}
g(\pi_{[\ell]}(S))  \geq  (1 - (1 +  \frac{\gamma^{3}_{g}}{\ell})^{-\ell} )\cdot g(\mathrm{OPT})+ [\Phi_{g}(S)-\frac{1}{\ell}\sum_{j\in[\ell]}\sum_{(u, k) \in S} \Phi_{g}(S-(u, k)+h_{j}((u,k))) ].
\end{align*}
The proof is complete.
\end{proof}

\begin{lemma}\label{ratiolemma}
Given an additive function $l(\cdot): E\rightarrow \mathbb{R}^{+}$, the local maximum $S$ found by non-oblivious local search algorithm guided by 
\begin{align*}
\Phi_{f}(S) = \Phi_{g}(S) + l(\pi_{[\ell]}(S)),
\end{align*}
satisfies that
\begin{align*}
g(\pi_{[\ell]}(S))+l(\pi_{[\ell]}(S))\geq (1 - (1 +  \frac{\gamma^{3}_{g}}{\ell})^{-\ell} ) \cdot g(\mathrm{OPT}) + l(\mathrm{OPT}).
\end{align*}
\end{lemma}
\begin{proof}
Note that
\begin{align*}
l(\pi_{[\ell]}(S)) = l(\mathrm{OPT}) + l(\pi_{[\ell]}(S)) -\sum_{(u, k) \in S} l(\pi_{[\ell]}(S-(u, k)+h_{j}((u,k))), \forall j\in [\ell],
\end{align*}
hence
\begin{align}\label{linearpart}
l(\pi_{[\ell]}(S)) = l(\mathrm{OPT}) + l(\pi_{[\ell]}(S)) -\frac{1}{\ell}\sum_{j\in[\ell]}\sum_{(u, k) \in S} l(S-(u, k)+h_{j}((u,k))), \forall j\in [\ell].
\end{align}
By summing inequalities (\ref{submodularpart}) and (\ref{linearpart}), we can finish the proof. 
\end{proof}

Let $l(\cdot) = (1-\kappa_{f})f(\cdot)$, $g(\cdot)=f(\cdot)-l(\cdot)$, we can obtain Theorem \ref{maintheorem} from Lemma \ref{ratiolemma}. 
\paragraph{Guessing the submodular ratio.} Note that for non-submodular $g(\cdot)$, we can geometrically guess the value of submodular ratio $\gamma_{g}$ when using $\Phi_{g}(\cdot)$, which is similar as \cite{harshaw2019submodular}.

\bibliographystyle{plain}	
\bibliography{paper}
\end{document}